\documentclass[12pt, journal,  onecolumn]{IEEEtran}
\usepackage[margin=1in]{geometry}
\usepackage{setspace}
\usepackage{amsthm,mathrsfs}
\usepackage{makecell}
\usepackage{multirow}
\usepackage[T1]{fontenc}
\usepackage{cite}
\usepackage{enumerate}

\usepackage{enumitem}
\usepackage[utf8]{inputenc} 
\usepackage[T1]{fontenc}
\usepackage{url}              
\usepackage{cite}             

\usepackage[cmex10]{amsmath}  
\usepackage{mleftright}       
\mleftright                   

\usepackage{graphicx}         
\usepackage{mathtools}

\usepackage{amsfonts}
\usepackage{bm}
\usepackage{bbm}
\usepackage{makecell}
\usepackage{multirow}
 \usepackage{amssymb}
\usepackage{txfonts}
\usepackage[T1]{fontenc}
\usepackage{tikz}
\usepackage[scr=dutchcal]{mathalfa}
\usepackage{ textcomp }
\usepackage{floatflt}
\usepackage{amsthm}                              
\usepackage{booktabs}

\theoremstyle{definition}
\newtheorem{Theorem}{Theorem}

\newtheorem{Corollary}{Corollary}

\newtheorem{Remark}{Remark}
\newtheorem{Definition}{Definition}

\DeclareMathOperator*{\argmax}{arg\,max}
\DeclareMathOperator*{\argmin}{arg\,min}

\IEEEoverridecommandlockouts
\begin{document}

\title{The Privacy-Utility Tradeoff in Rank-Preserving Dataset Obfuscation} 


\author{%

\IEEEauthorblockN{ Mahshad Shariatnasab, Farhad Shirani, S. Sitharma Iyengar\thanks{This work was supported in part by NSF grants CCF-2241057.}}
\IEEEauthorblockA{\\Florida International University
\\Email: mshar075@fiu.edu, fshirani@fiu.edu, iyengar@fiu.edu}
}

\maketitle

\begin{abstract}
 Dataset obfuscation refers to techniques in which random noise is added to the entries of a given dataset, prior to its public release, to protect against leakage of private information. In this work, dataset obfuscation under two objectives is considered: i) rank-preservation: to preserve the row ordering in the obfuscated dataset induced by a given rank function, and ii) anonymity: to protect user anonymity under fingerprinting attacks. The first objective, rank-preservation, is of interest in applications such as the design of search engines and recommendation systems, feature matching, and social network analysis. Fingerprinting attacks, considered in evaluating the anonymity objective, are privacy attacks where an attacker  constructs a fingerprint of a victim based on its observed activities, such as online web activities, and compares this fingerprint with information extracted from a publicly released obfuscated dataset to identify the victim.
By evaluating the performance limits of a class of obfuscation mechanisms over asymptotically large datasets,
a fundamental trade-off is quantified between rank-preservation and user anonymity. 
Single-letter obfuscation mechanisms are considered, where each entry in the dataset is perturbed by independent noise, and their fundamental performance limits are characterized by leveraging large deviation techniques. The optimal obfuscating test-channel, optimizing the privacy-utility tradeoff, is characterized in the form of a convex optimization problem which can be solved efficiently. 
Numerical simulations of various scenarios are provided to verify the theoretical derivations.
\end{abstract}

\section{Introduction}

Dataset privacy is a major concern due to the potential risks associated with the misuse of personal and sensitive information included in various datasets. If the data to be released has no immediate utility,
then cryptographic methods suffice to preserve privacy \cite{wang2019privacy,salomon2003data}. However, when data is released publicly for a specific immediate utility --- such as the release of \textit{anonymized} social network data to advertising companies --- the  necessarily unencrypted disclosure incurs a privacy
risk and may lead to unwanted inferences \cite{martin2017role,narayanan2009anonymizing,takbiri2018matching, calandrino2011you, shariatnasab2022fundamental ,wondracek2010practical}. Obfuscation provides a mitigating solution, by introducing noise in the dataset entries prior to their release. This leads to a privacy-utility tradeoff, where increased perturbation of the dataset entries via random noise leads to increased privacy at the expense of lost utility. 
In this work, we study this fundamental privacy-utility tradeoff and characterize optimal obfuscation strategies, where  privacy is evaluated under fingerprinting attacks \cite{takbiri2018matching, calandrino2011you,shariatnasab2022fundamental,wondracek2010practical}, and utility is measured via metrics associated with rank-preservation \cite{jeong2022ranking, alabi2022private, yan2020private, shang2014application}. 


Obfuscation mechanisms protect privacy via noisy perturbations of the dataset entries. A widely studied class of obfuscation mechanisms is to perturb each dataset entry independently by passing them through identical test-channels \cite{takbiri2018matching, shi2016privacy, zurbaran2015near, wightman2011evaluation}. We call these   mechanisms \textit{single-letter obfuscation mechanisms} since their operations can be characterized using single-letter conditional probability measures. Single-letter obfuscation mechanisms, as opposed to multi-letter mechanisms, 
are amiable to analysis, and they have good performance under specific utility metrics such as the variational distance and Euclidean distance metrics \cite{zamani2021data, rassouli2019optimal, sankar2013utility, liao2019tunable}. Furthermore, perturbation via independent noise reduces information leakage among entries of the obfuscated dataset. Consequently, in this work, we focus our study to single-letter obfuscation mechanisms and their fundamental performance limits. 


Rank-preservation is a utility metric of interest in dataset obfuscations \cite{jeong2022ranking, alabi2022private, yan2020private, shang2014application, dwork2001rank, hay2017differentially}. In general, 
for a given dataset $\mathsf{X}$ with $n\in \mathbb{N}$ rows, a rank function $R:[n]\to [n]$ is a mapping which assigns an ordering to the rows of the dataset. For instance, let us consider a social network with $n\in\mathbb{N}$ users, and let $\mathsf{X}=[X_{i,j}]_{i,j\in [n]}$ be the adjacency matrix capturing the user's connections in the social network, where $X_{i,j}=1$ if the $i$th and $j$th users are connected, and $X_{i,j}=0$ otherwise. The user-degree-based rank function induces an ordering of the users based on number of connections, i.e. $R(i)<R(i')$ if $\sum_{j\in [n]}X_{i,j}<\sum_{j\in [n]}X_{i',j}$. Rank functions are used in the design of search engines, social network analysis, feature matching, and recommendation systems \cite{liu2021survey, brancotte2015rank, jeong2022ranking}. Rank recovery algorithms reconstruct the rank function associated with a given dataset based on noisy observations, e.g., by observing an obfuscated dataset. That is, given an obfuscated dataset $\mathsf{Y}$, a rank-recovery algorithm produces a reconstruction $\widehat{R}(\cdot)$ of the rank function $R(\cdot)$ associated with the original dataset $\mathsf{X}$. The performance of the rank recovery algorithm is measured with respect to an underlying  distortion metric, measuring the distance between the original and recovered rank-functions. A widely used distortion metric, considered in this work, is the Kendall's rank correlation coefficient (KRCC) \cite{alabi2022private, brancotte2015rank, jeong2022ranking, zamani2021data}. The KRCC distance $d(R,\widehat{R})$ counts the number of pairwise disagreements between the two rank functions, i.e. $d(R,\widehat{R})\triangleq \sum_{i\in [n]}\mathbbm{1}(R(i)\neq \widehat{R}(i))$. 


We study the privacy-utility tradeoff in database obfuscation, where the utility objective is rank-preservation discussed in the prequel, and privacy is evaluated under fingerprinting attacks. Fingerprinting attacks are a major threat to users' privacy in social networks, mobility networks, and wireless networks, among others \cite{takbiri2019asymptotic, su2017anonymizing, domingo2016database}. In these attacks, given an obfuscated dataset, the attacker's objective is to identify the row in the dataset corresponding to a victim by acquiring a partial fingerprint of the victim's real-world activities, comparing it  with each of the rows in the obfuscated dataset, and detecting the row with correlated entries (Figure \ref{fig:model}). To provide an example, let us consider online fingerprinting attacks which rely on social network group memberships \cite{wondracek2010practical, calandrino2011you, shariatnasab2022fundamental}. In such scenarios, an attacker controls a malicious website, the victim is a visitor to the website,  and the attacker uses 
 browser history sniffing techniques to extract a partial list of social network groups visited by the victim \cite{wondracek2010practical, narayanan2009anonymizing}. The extracted information can be represented by a binary vector $F^q=(F_1,F_2,\cdots,F_q), q\in \mathbb{N}$, where $F_i=1$ if the victim has visited the $i$th social network group's website, and $F_i=0$, otherwise. The vector  $F^q$ is called the fingerprint of the victim. To identify the victim's social network account, the attacker scans the social network and acquires a (obfuscated) dataset $\mathsf{Y}$ capturing the public group memberships in the social network. It then compares the fingerprint $F^q$ and the dataset $\mathsf{Y}$ to find the closest match and identify the victim. 
 In practice, the attacker acquires each fingerprint element by \textit{querying} the user's activities, and there is a cost associated with each query. For instance, in social network fingerprinting attacks described above, the state-of-the-art browser history sniffing techniques can make between tens to several thousand queries per second depending on the victim's device and web browser \cite{smith2018browser,solomos2021tales}. So, the cost associated with each fingerprint element is the time spent to query the value of that element using browser history sniffing. As a result, the length of the partial fingerprint is determined by the attacker's resources. In this work, the privacy objective under consideration is to minimize the information leakage about the victim's identity given a partial fingerprint $F^q$ with a fixed length $q\in \mathbb{N}$. 
\\The following is a summary of our contributions: 
 \begin{itemize}[leftmargin=*]
     \item To formulate the dataset obfuscation problem under the rank-preservation and anonymity constraints. 
     \item To evaluate the fundamental performance limits of single-letter obfuscation mechanisms and quantify a tradeoff between the two objectives. This allows the system designer to choose the appropriate amount of obfuscation through the choice of a single-letter test-channel by optimizing the aforementioned trade-off. 
     \item To characterize the optimal obfuscating test-channel, optimizing the privacy-utility tradeoff, in the form of a convex optimization problem.
     \item To provide numerical simulations under various statistical scenarios. 
 \end{itemize}

\begin{figure}[t]
 \centering \includegraphics[width=0.7\linewidth, draft=false]{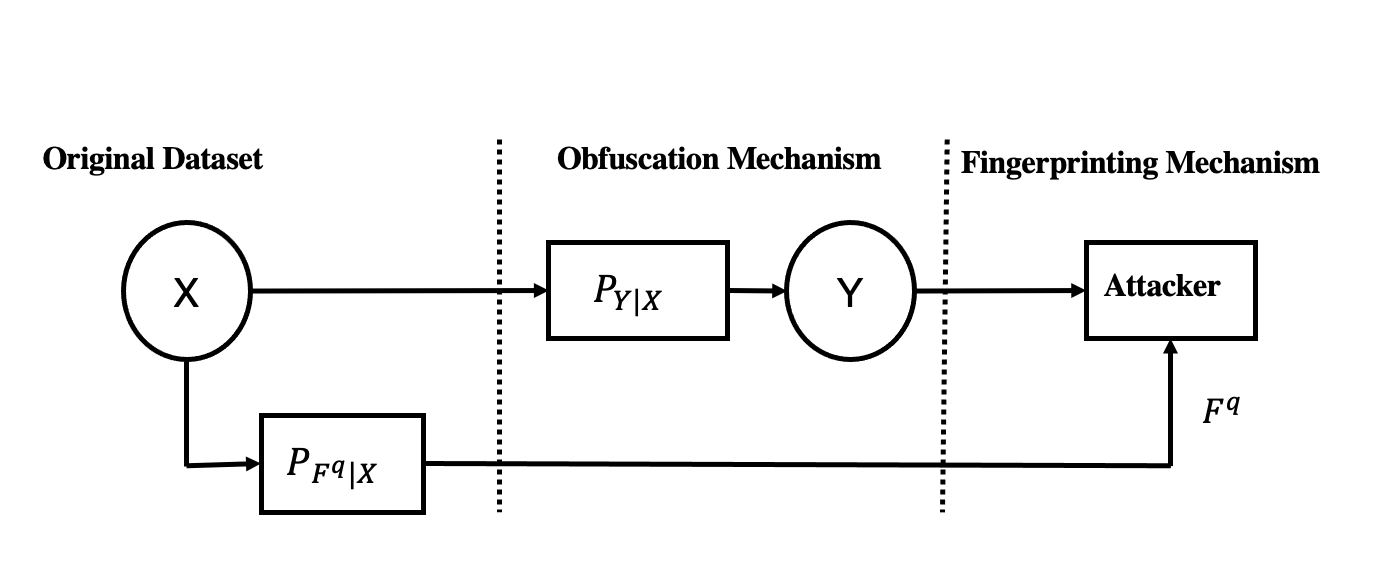}
 \caption{$\mathsf{X}$ represents the original dataset. $\mathsf{Y}$ represents the obfuscated dataset. The attacker acquires the fingerprint vector $F^q$ by querying the victim's activities and attempts to identify the victim by comparing $F^q$ and $\mathsf{Y}$.}
 \label{fig:model}
\end{figure}

\textit{Notation:} The random variable $\mathbbm{1}_{\mathcal{E}}$ is the indicator of the event $\mathcal{E}$.
 The set $\{n,n+1,\cdots, m\}, n,m\in \mathbb{N}$ is represented by $[n,m]$, and for the interval $[1,m]$, we use the shorthand notation $[m]$. 
 For a given $n\in \mathbb{N}$, the $n$-length vector $(x_1,x_2,\hdots, x_n)$ is written as $x^n$, $\underline{x}$, and $(x_i)_{i\in [n]}$, interchangably. The notation $[x_{i,j}]_{i\in [n],j_[m]}$ denotes an $n\times m$ matrix, where $x_{i,j}$ is the element on the $i$th row and $j$th column. We use sans-serif letter such as $\mathsf{X}$ and $\mathsf{x}$ to represent matrices.

\vspace{-0.05in}
\section{Problem formulation}
\label{sec:page-limit}

\label{sec1}
\vspace{-0.05in}
In this section, we describe the mathematical formulation of the dataset obfuscation problem shown in Figure \ref{fig:2}.
\\\textbf{Random Dataset:} A dataset 
is a matrix $\mathsf{X}=[x_{i,j}]_{i\in [n], j \in [m]}$, where $x_{i,j} \in \mathcal{X}$, the set $\mathcal{X}$ is finite, and $n,m\in \mathbb{N}$. Each row $x_i^m=(x_{i,1}, x_{i,2}, x_{i,3}, \cdots x_{i,m}), i\in[n]$ is called an entry of the dataset, $m \in \mathbb{N}$ is the length of the entries, $n \in \mathbb{N}$ is the size of the dataset. The dataset is said to have $n$ members. 
We consider stochastically generated datasets with independent and identically distributed (IID) elements, where given a distribution $P_X$ defined on alphabet $\mathcal{X}$, we have:
\[P(\mathsf{X}=[x_{i,j}]_{i\in [n], j \in [m]})= \prod_{i\in [n],j\in [m]}P_X(x_{i,j})\]
A random dataset is parameterized by $(n,m,\mathcal{X},P_X)$.
\\\textbf{Original and Obfuscated Datasets:} An agent, Alice, has access to an original dataset $\mathsf{X}$ parameterized by $(n,m,\mathcal{X},P_X)$. Alice wishes to disclose an obfuscated dataset $\mathsf{Y}=f(\mathsf{X})$ to Bob, where $f(\cdot)$ is a possibly stochastic function captured by $P_{\mathsf{Y}|\mathsf{X}}$. Bob's objective is to recover the row-ordering of the original dataset, with respect to a given rank function $R(\cdot)$, by leveraging the obfuscated dataset. The rank function and privacy constraints under consideration are described in more detail in the sequel. 
\\\textbf{Privacy Objective:} An attacker, Eve, gains access to the disclosed dataset $\mathsf{Y}$. Eve's objective is to identify the dataset entry corresponding to a specific victim. To elaborate, we let $U$ represent the row index corresponding to the victim of interest. The random variable $U$ is assumed to be uniformly distributed on $[n]$. Eve acquires a partial fingerprint $F^q$ of the row elements  $(X_{U,i_1},X_{U,i_2},\cdots, X_{U,i_q})$ corresponding to the victim, where 
\begin{align*}
P(F^q=f^q|   (X_{U,i_j})_{j\in [q]}=x^q) =\prod_{i=1}^q P_{F|X}(f_i|x_{i_j}), f^{q}, x^{q}\in \mathcal{F}^q\times \mathcal{X}^q,
\end{align*}
and $P_{F|X}$ is a collection of $|\mathcal{X}|$ probability measures defined on a finite set $\mathcal{F}$. The fingerprint vector and obfuscated dataset are conditionally independent of each other given the original dataset, i.e.,  the Markov chain $F_i\! \!\leftrightarrow \!\! X_{U,i_j} \!\!\leftrightarrow\!\! Y_{U,i_j}, i_j\in [m], j\in [q]$ holds. 
One of the objectives in the dataset obfuscation problem is to minimize the information leakage between the victim's row index and Eve's observations $(\mathsf{Y},F^q)$. That is to minimize $I(U;\mathsf{Y},F^q)$. We assume that the fingerprinting process is unsupervised in the sense that Eve does not have a choice on which indices $i_j, j\in [q]$ are queried to extract the fingerprint.
\\\textbf{Rank-Preservation Objective:} In general, given a dataset $\mathsf{X}$ a rank function is a mapping $R: [n] \to [n]$ which induces an ordering on the rows in the dataset, i.e.,  $R(i)$ indicates the rank of the $i$th row of $\mathsf{X}$ induced by the rank function $R(\cdot)$.  In this work, we consider the degree-based rank function defined in the following. The degree-based rank function and its variants are used in applications such as social network analysis, search engine design, and recommendation systems \cite{chung2014brief, berkhin2005survey,voudigari2016rank}.

\begin{Definition}[\textbf{Degree-Based Rank Function}]
\label{def1}
Given a dataset $\mathsf{X}=[x_{i,j}]_{i\in [n],j\in [m]}$, the degree of the $i$th row is defined as ${D}(i)\triangleq \sum_{j=1}^mx_{i,j}$. The degree-based rank function $R_d: [n] \to [n]$  is characterized by the following relation:
\[\forall i,i'\in [n],i<i': R_d(i)\leq  R_d(i') \iff D(i)\leq D(i').\]
 \end{Definition}
\begin{Remark}
In this work, we have considered a degree-based rank function which does not discriminate between different elements of each row in calculating the degree. The analysis can potentially be extended to weighted-degree-based rank functions, where the degree is computed as a weighted sum of the row elements, i.e, ${D}(i)\triangleq \sum_{j=1}^mw_jx_{i,j}, w_j\in \mathbb{R}, i\in [n]$. 
\end{Remark}
Bob receives the obfuscated dataset $\mathsf{Y}$, and wishes to reconstruct the rank function $R(\cdot)$ associated with $\mathsf{X}$. We consider the conventionally used KRCC (e.g., see \cite{kendall1938new}) as the distortion criterion measuring the quality of Bob's reconstructed rank function $\widehat{R}(\cdot)$.
\begin{Definition}[\textbf{Kendall Rank Correlation Coefficient}]
\label{def:KRCC}
For two rank functions $R_d(\cdot)$ and $\widehat{R}_d(\cdot)$, their KRCC distance is defined as\footnote{In some texts KRCC is defined as $d'_{KRCC}(R_d,\widehat{R}_d)\triangleq\frac{2}{n(n-1)}\sum_{i<j}sign(R_d(i)-R_d(j))sign(\widehat{R}_d(i)-\widehat{R}_d(j))$. It can be observed that $d'_{KRCC}(\cdot,\cdot)=1-\frac{n}{(n-1)}d_{KRCC}(\cdot,\cdot)$. We adapt the formulation in Definition \ref{def:KRCC} as it allows for more concise arguments.}
\begin{align}
   & d_{KRCC}(R_d,\widehat{R}_d)\triangleq \frac{1}{n(n-1)}\!\!\!\sum_{(i,j)\in [n]}\!\!\! \mathbbm{1}\Big(R_d(i)>{R}_d(j) \!\!\And\!\! \widehat{R}_d(i)\!<\!\widehat{R}_d(j)\Big)
\end{align}
\end{Definition}

\begin{figure}[t]
 \centering \includegraphics[width=0.6\linewidth, draft=false]{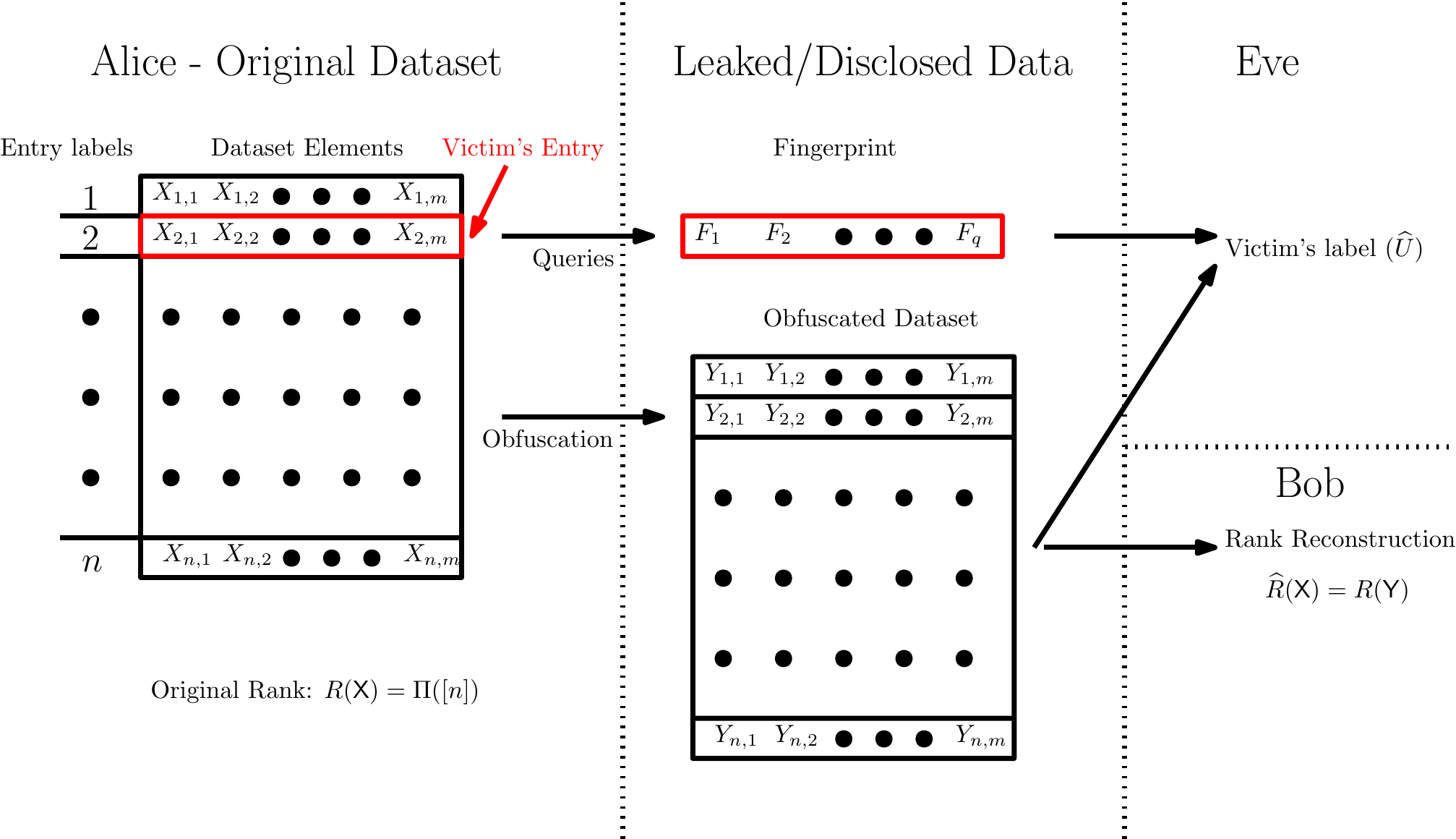}
 \caption{The dataset obfuscation setup.}
 \label{fig:2}
\end{figure}





\noindent \textbf{k-Letter Obfuscation Strategy:} As mentioned in the introduction, a widely used obfuscation method is to perturb each dataset entry independently by passing them through identical test-channels \cite{takbiri2018matching, shi2016privacy, zurbaran2015near, wightman2011evaluation}. We call such mechanisms  {single-letter obfuscation mechanisms}. One justification for their use is that in applications such as search engines and recommendation systems, standard ranking algorithms such as PageRank \cite{chung2014brief} require both an accurate estimation of the degree-based rank function and a small $\ell_1$ distance between the original dataset and the obfuscated dataset for reliable performance. Single-letter obfuscation mechanisms facilitate analyzing and controlling the $\ell_1$ distance between the two datasets by appropriate choice of the underlying obfuscating test-channel. A k-letter obfuscation strategy is a generalization of single-letter strategies, where randomly partitioned subsets of size $k$ of elements of each entry are passed through $k$-letter test-channels for obfuscation. 
The following formally defines a k-letter obfuscation strategy. 
\begin{Definition}[\textbf{k-letter Obfuscation Strategy}]
\label{def:k_lett}
For a random dataset $\mathsf{X}=[X_{i,j}]_{i\in [n],j\in [m]}$ parametrized by $(n,m,\mathcal{X},P_X)$, a k-letter obfuscation strategy  is parametrized by the conditional distribution $P_{Y^k|X^k}$. The obfuscated dataset $\mathsf{Y}=[Y_{i,j}]_{i\in [n],j\in [m]}$ is produced as follows:\footnote{For ease of notation, we have assumed that $m$ is divisible by $k$.}
\begin{align*}
    P_{\mathsf{Y}|\mathsf{X}}(\mathsf{y}|\mathsf{x})= \prod_{i\in [n]}\prod_{\ell\in [\frac{m}{k}]} P_{Y^k|X^k}((y_{i,j})_{j\in \mathcal{P}_\ell}| (x_{i,j})_{j\in \mathcal{P}_\ell}),
\end{align*}
where $\mathsf{x}=[x_{i,j}]_{i\in [n],j\in [m]}$, $\mathsf{y}=[y_{i,j}]_{i\in [n],j\in [m]}$ and
$(\mathcal{P}_{\ell})_{\ell\in [\frac{m}{k}]}$ is a randomly and uniformly chosen partition of $[m]$ into $\frac{m}{k}$ subsets of size equal to $k$. 
\end{Definition}

The dataset obfuscation problem is formally defined in the following. 
\begin{Definition}[\textbf{k-Letter Dataset Obfuscation Problem}]
\label{def:problem}
Given a random dataset parametrized by $(n,m,\mathcal{X},P_X)$, fingerprint length $q\in \mathbb{N}$, query noise distribution $P_{F|X}$, and $\epsilon>0$, the k-letter dataset obfuscation problem is to characterize the $\epsilon$-optimal k-letter strategy $P^*_{Y^k|X^k}$, defined as
\begin{align*}
    &P^*_{Y^k|X^k}\triangleq \argmin_{P_{Y^k|X^k}:\frac{1}{q}I(U; F^q,\mathsf{Y})<\epsilon} \quad \mathbb{E}(d_{KRCC}(R_d, \widehat{R}_d)),
\end{align*}
where $R$ and $\widehat{R}$ are the degree-based rank functions associated with $\mathsf{X}$ and $\mathsf{Y}$, respectively, and $U$ is uniformly distributed over $[n]$. The set of all pairs $(\epsilon,\delta)$ for which there exists $P_{Y^k|X^k}$ such that $\mathbb{E}(d_{KRCC}(R, \widehat{R}))<\delta$ and $\frac{1}{q}I(U; F^q,\mathsf{Y})<\epsilon$ is called the \textit{achievable privacy-utility set} and is denoted by $\mathcal{R}(k,n,m,q,\mathcal{X},P_X,P_{F|X})$.
\end{Definition}

In the rest of the paper, for brevity, we denote the achievable privacy-utility region by  $\mathcal{R}(n,m,P_X,P_{F|X})$ when the values of $k,q$ and $\mathcal{X}$ are clear from the context.

\section{Characterizing the Privacy-Utility Tradeoff}
\label{sec:th}
In this section, we consider single-letter obfuscation mechanisms and evaluate their fundamental performance limits, in terms of the utility-privacy tradeoff measured with respect to KRCC utility metric and information leakage privacy metric described in the previous section. The analysis can also be extended to finite-letter obfuscation mechanism using similar techniques. For ease of explanation, the main theorems are provided for binary alphabet datasets. 

\label{sec:1}
Recall that given a dataset $\mathsf{X}$, parametrized by $(n,m,\mathcal{X},P_X)$, and a conditional distribution $P_{Y|X}$, a single-letter obfuscation mechanism produces the obfuscated dataset $\mathsf{Y}$ conditioned on  $\mathsf{X}$ by passing each element of $\mathsf{X}$ through independent and statistically identical  test-channels characterized by $P_{Y|X}$. In order to provide our main results, we first introduce the following notation. Given  joint distribution $P_{X,Y}=P_XP_{Y|X}$ on pairs of binary variables $(X,Y)$, we define $Q(P_X,P_{Y|X})\triangleq F_{N_1,N_2}(0,0)$, where $F_{N_1,N_2}(\cdot,\cdot)$ is the cumulative distribution function (CDF) of zero-mean jointly Gaussian variables $N_1$ and $N_2$ with covariance matrix $\Sigma\triangleq [\sigma_{i,j}]_{i,j\in \{1,2\}}$ given by
 \begin{align}
    & \sigma_{1,1}\triangleq 2P_X(0)P_X(1), \qquad \sigma_{2,2}\triangleq 2P_Y(0)P_Y(1), \qquad \label{eq:gaus}
     \sigma_{2,1}=\sigma_{1,2}\triangleq 2P_X(1)(P_Y(1)- P_{Y|X}(1|1)).
\end{align}   
The following provides one of the main results of the paper.
\begin{Theorem}
\label{th:1}
Let $q,n,m\in \mathbb{N}, \mathcal{X}=\mathcal{F}=\{0,1\}$, $P_X$ be a probability measure on $\mathcal{X}$, and $P_{F|X}$ a collection of probability measures on $\mathcal{F}$. Then,
there exists $b>0$ such that:
\begin{align*}
&\mathcal{R}_{in}(n,m,P_X,P_{F|X})\subseteq\mathcal{R}(n,m,P_X,P_{F|X})\subseteq\mathcal{R}_{out}(n,m,P_X,P_{F|X}),
\end{align*}
where
\begin{align*}
&\mathcal{R}_{in}(n,m,P_X,P_{F|X})\triangleq
\bigcup_{P_{Y|X}}
\Big\{(\epsilon,\delta)| \epsilon\geq I(P_Y;P_{Y|F})+\!\zeta\!+b\frac{\log^{\frac{3}{2}}m}{\sqrt{m}}, 
\delta \geq Q(P_{X},P_{Y|X})+\frac{(42\sqrt[4]{2}+16)}{\sqrt{m}}\theta\gamma
\Big\},
\\
&\mathcal{R}_{out}(n,m,P_X,P_{F|X})\triangleq
\bigcup_{P_{Y|X}}
\Big\{(\epsilon,\delta)| \epsilon\geq I(P_Y;P_{Y|F}), 
\delta \geq Q(P_{X},P_{Y|X})-\frac{(42\sqrt[4]{2}+16)}{\sqrt{m}}\theta\gamma
\Big\},
\\& \theta\triangleq \frac{4}{\sqrt{\lambda^*}},\qquad 
\lambda^*\triangleq \min\{\sigma_{1,1}-|\sigma_{1,2}|, \sigma_{2,2}-|\sigma_{2,1}|\},
\\&\gamma\triangleq  2P(X=Y)P(X\neq Y)+
2^{\frac{5}{2}}(P_{X,Y}(0,0)P_{X,Y}(1,1)+P_{X,Y}(0,1)P_{X,Y}(1,0)),
\\&\zeta\triangleq \max(\max_{P_Y}(I(P_Y;P_{Y|F})-\frac{\log{n}}{q}),0),
\end{align*}
 the mutual information $I(P_Y;P_{Y|F})$ is evaluated with respect to $P_{Y,F}$ induced by the Markov chain $Y\leftrightarrow X\leftrightarrow F$, and the union is over all probability distributions $P_{Y|X}$. 
Particularly, for asymptotically large datasets, we have:
\begin{align*}
   & \!\!\lim_{m\to \infty}\!\mathcal{R}(n,m,\!P_X\!,\!P_{F|X})\!=\!
 \bigcup_{P_{Y|X}}  \{(\epsilon,\delta)| \epsilon\!\geq \!I(P_Y;P_{Y|F}), \delta\! \geq \! Q(P_{X},P_{Y|X})\}.
\end{align*}
\end{Theorem}
\begin{proof}
Please refer to Appendix \ref{A:1}.
\end{proof}
Theorem \ref{th:1} provides the achievable utility-privacy  region as a union of achievable regions for each obfuscating test-channel $P_{Y|X}$. A relevant problem of interest is to find the optimal test channel $P^\epsilon_{Y|X}$ minimizing the utility cost $\delta$ given a privacy cost $\epsilon$. The following theorem provides a characterization of $P^\epsilon_{Y|X}$ in the form of a computable convex optimization problem for asymptotically large datasets, i.e., for $m\to \infty$.

\begin{Theorem}
\label{th:2}
    Let $q,n,m\!\in\! \mathbb{N}, \mathcal{X}\!=\!\mathcal{F}\!=\!\{0,1\}$, $P_X$ be a probability measure on $\mathcal{X}$, and $P_{F|X}$ be a collection of probability measures on $\mathcal{F}$, such that $
        \max_{P_Y} I(P_Y;P_{Y|F})\leq \frac{\log{n}}{q}$.  Define:
    \begin{align*}
        P^\epsilon_{Y|X}\triangleq \argmin_{P_{Y|X}:I(P_Y;P_{Y|F})\leq \epsilon} \{ \delta| (\epsilon,\delta)\in \lim_{m\to \infty}\mathcal{R}(n,m,P_X,P_{F|X})\}, \epsilon>0.
    \end{align*}
Then, 
    \begin{align}
        P^\epsilon_{Y|X}=
     &\argmin_{P_{Y|X}:I(P_Y,P_{Y|F})=\epsilon}   \quad \frac{\text{Cov}(N_1,N_2)}{\sqrt{\text{Var}(N_1)\text{Var}(N_2)}},
                \label{eq:th2}
    \end{align}
    where $N_1$ and $N_2$ are zero-mean jointly Gaussian variables with covariance matrix given in Equation \eqref{eq:gaus}.
\end{Theorem}
\begin{proof}
    Please refer to Appendix \ref{A:3}.
\end{proof}
The following follows from the proof of Theorem \ref{th:2}.
\begin{Corollary}
\label{cor:1}
The optimal obfuscating test-channel in Equation \eqref{eq:th2} can be computed through the following optimization:
    \begin{align}
        P^\epsilon_{Y|X}=
&\!\!\!\!\!\argmax_{\substack{p_1,p_2:I(P_Y,P_{Y|F})=\epsilon\\{p_1+p_2\leq 1}}} \!   \frac{(p_1+p_2-1)^2}{(\overline{p}_1P_X(0)\!+\!p_2P_X(1))(p_1P_X(0)\!+\!\overline{p}_2P_X(1))},
     \label{eq:th:3}
    \end{align}
    where $\overline{p}_i\triangleq 1-p_i, i=1,2$, and $P_{Y|X}(1|0)=p_1, P_{Y|X}(0|1)=p_2$. 
\end{Corollary}
We show that the objective function in the optimization in Equation \eqref{eq:th:3} is convex. To see this, let us define $a\triangleq p_1+p_2$ and $b\triangleq 2P_X(0)p_1-2P_X(1)p_2+1-2P_X(0)$. Then, the objective function can be written as 
\begin{align*}
    g(a,b)\triangleq \frac{4(a-1)^2}{(1-b)(1+b)}= \frac{4(a-1)^2}{1-b^2}, a\in [0,1], b\in [-1,1].
\end{align*}
The Hessian matrix of second partial derivatives of $\frac{1}{4}g(a,b)$ is given as:
\begin{align*}
  \mathcal{H}=  \begin{bmatrix}
        \frac{\partial^2}{\partial a^2}\frac{1}{4}g(a,b)&  \frac{\partial^2}{\partial a\partial b}\frac{1}{4}g(a,b)\\
        \frac{\partial^2}{\partial a\partial b}\frac{1}{4}g(a,b)& \frac{\partial^2}{\partial b^2}\frac{1}{4}g(a,b)
    \end{bmatrix}=
\begin{bmatrix}
     \frac{2}{1-b^2}&  \frac{4(a-1)b}{(1-b^2)^2}\\
         \frac{4(a-1)b}{(1-b^2)^2}& \frac{2(a-1)^2(1+3b^2)}{(1-b^2)^3}
    \end{bmatrix}
\end{align*}
We have:
\begin{align*}
    det(\mathcal{H})= \frac{4(a-1)^2(1-b^2)}{(1-b^2)^4}\geq 0.
\end{align*}
So, $g(a,b)$ is a convex function, and since $(a,b)$ are a linear transformation of $(p_1,p_2)$, the objective function in Equation \eqref{eq:th:3} is convex in $(p_1,p_2)$. This optimization problem can in general be solved efficiently using numerical methods. There are special cases where exact analytical solution can be derived. For instance, the following corollary characterizes the optimal test-channel if the query noise $P_{F|X}$ is a binary symmetric channel with transition probability $q$ (BSC(q)), and the choice of the obfuscating test-channel is restricted to BSC test-channels, i.e. $P_{Y|X}(0|1)=P_{X|Y}(1|0)=p, p\in [0,1]$.

\begin{Corollary}
\label{cor:2}
   Assume that $P_{F|Y}$ is a BSC(q) channel, where $q\in [0,\frac{1}{2}]$, and the choice of obfuscating test-channel $P_{Y|X}$ is restricted to BSC test channels. Then, given $\epsilon>0$ the optimizing obfuscating test-channel, minimizing the KRCC cost is the BSC parametrized by 
   \[p= \frac{h^{-1}_b(1-\epsilon)-q}{1-2q},\]
   where $h^{-1}_b(\cdot)$ is the inverse of the binary entropy function defined as $h_b(x)=-x\log_2(x)-(1-x)\log_2(1-x),x\in [0,\frac{1}{2}]$.
\end{Corollary}

\section{Numerical Simulations}
\label{sec:sim}
\begin{figure}[t]
 \centering \includegraphics[width=0.6\linewidth, draft=false]{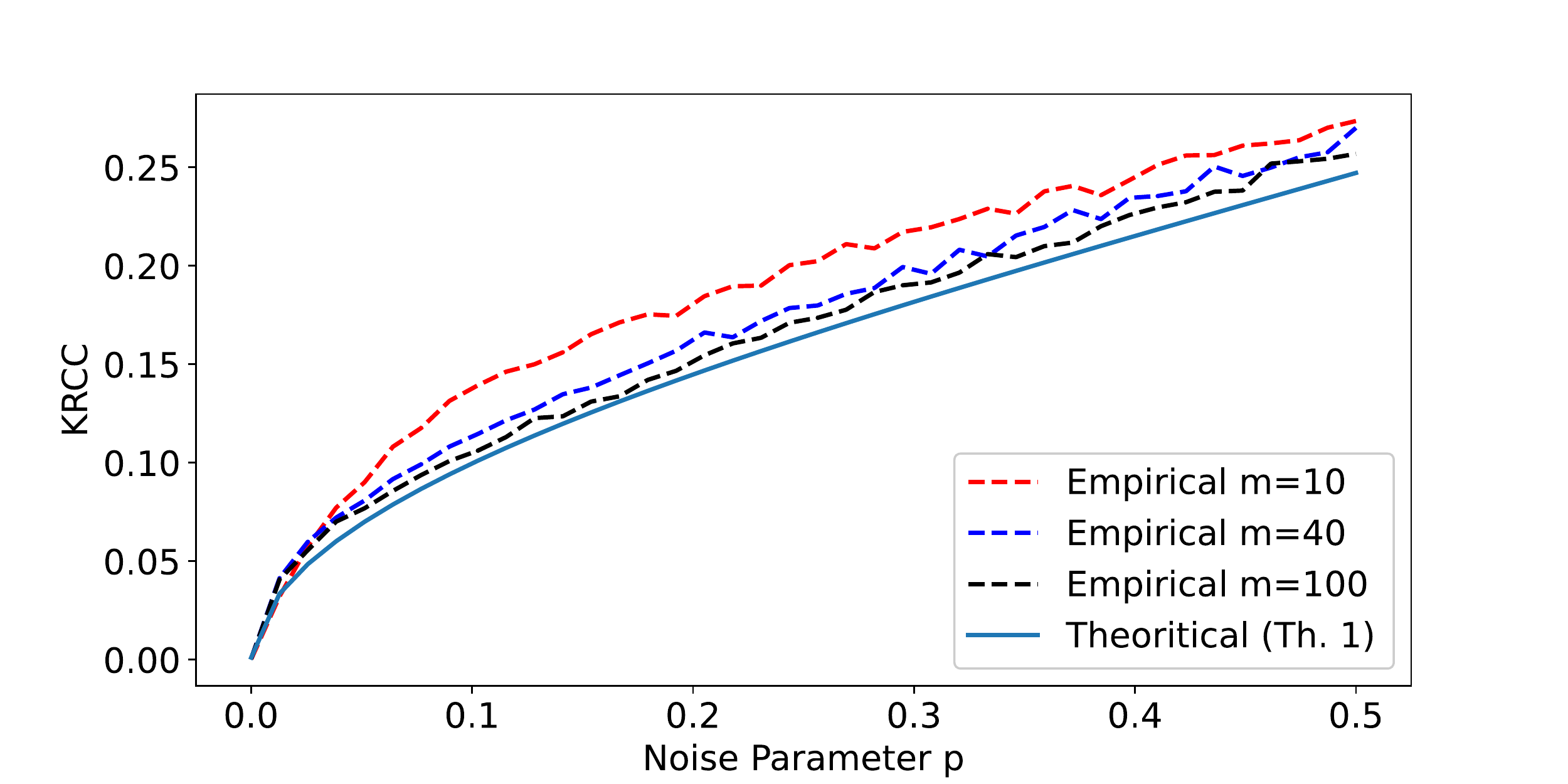}
 \caption{Comparison of analytical derivation of KRCC with empirical observations through numerical simulation.}
 \label{fig:1}
\end{figure}
This section provides numerical simulations to illustrate some of the theoretical derivations provided in the prior sections.
\subsection{Analytical and Empirical Simulation of KRCC}
In the proof of Theorem \ref{th:1}, we show that the resulting KRCC from obfuscating a dataset $\mathsf{X}$ parametrized by $(n,m,\mathsf{X},P_X)$ using a single-letter obfuscation mechanism $P_{Y|X}$ is given by $Q(P_X,P_{Y|X})$. To verify this, we have simulated the obfuscation mechanism when the obfuscating test-channel is $BSC(p), p\in [0,\frac{1}{2}]$ is applied to a dataset with $P_X(0)=P_X(1)=\frac{1}{2}$, $n=200$, and $m\in \{10,40,100\}$. To ensure accuracy, we have performed numerical simulations for each value of $m$ by generating the dataset 40 times, performing obfuscation and calculating the resulting KRCC. Figure \ref{fig:1} shows the resulting analytical and empirically observed KRCCs. As can be observed the analytical result is close to the empirical performance and the empirical KRCC converges to the analytical derivation as $m$ becomes larger. 
\subsection{Asymmetric Obfuscating Test-Channels}
\begin{figure}[t]
 \centering \includegraphics[width=0.6\linewidth, draft=false]{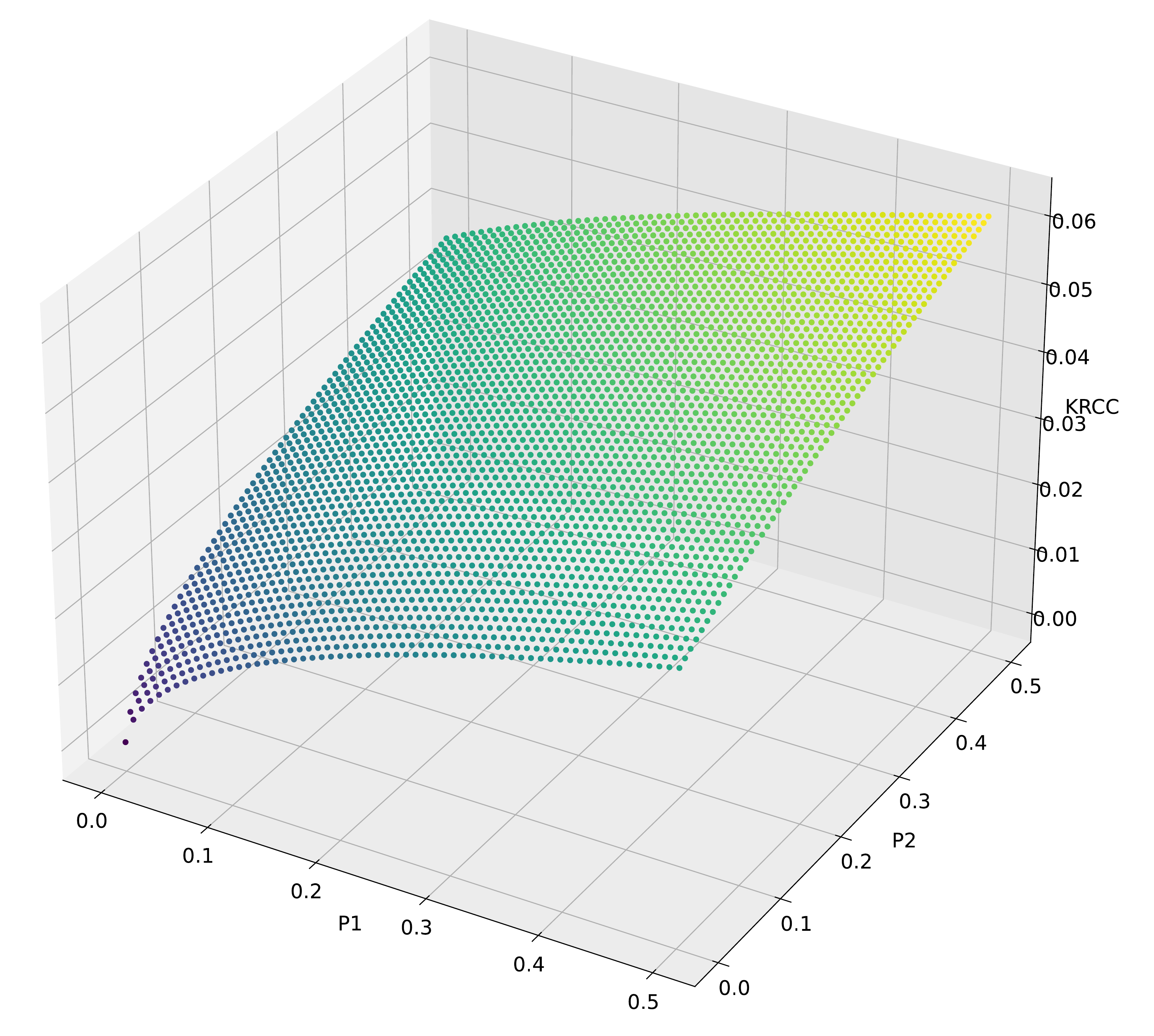}
 \caption{KRCC values for asymmetric obfuscation mechanisms.}
 \label{fig:4}
\end{figure}
In Section \ref{sec:th}, we argue that the objective function of Equation \eqref{eq:th:3} is convex which implies the KRCC is concave as a function of $p_1=P_{Y|X}(1|0)$ and $p_2=P_{Y|X}(0|1)$. Figure \ref{fig:4} shows the KRCC when $P_X(0)=P_X(1)=\frac{1}{2}$ and $m\to \infty$ for $p_1,p_2\in [0,\frac{1}{2}]$. It can be observed that KRCC is convex in $(p_1,p_2)$ as predicted. 
\subsection{Privacy-Utility Tradeoff}
In Figure \ref{fig:3}, we have shown the privacy-utility tradeoff for the scenario where a symmetric dataset ($P_X(1)=\frac{1}{2}$) is obfuscated using an optimal symmetric test-channel ($P_{Y|X}(0|1)=P_{Y|X}(1|0)$), and the query noise is modeled by a BSC(0.1). The resulting achievable privacy-utility region $\mathcal{R}(n,m,P_X,P_{F|X})$ is shown as the blue shaded region in the figure. The optimal symmetric test-channel used in the simulation is derived using Corollary \ref{cor:2} in the previous section. 
\begin{figure}[h]
 \centering \includegraphics[width=0.6\linewidth, draft=false]{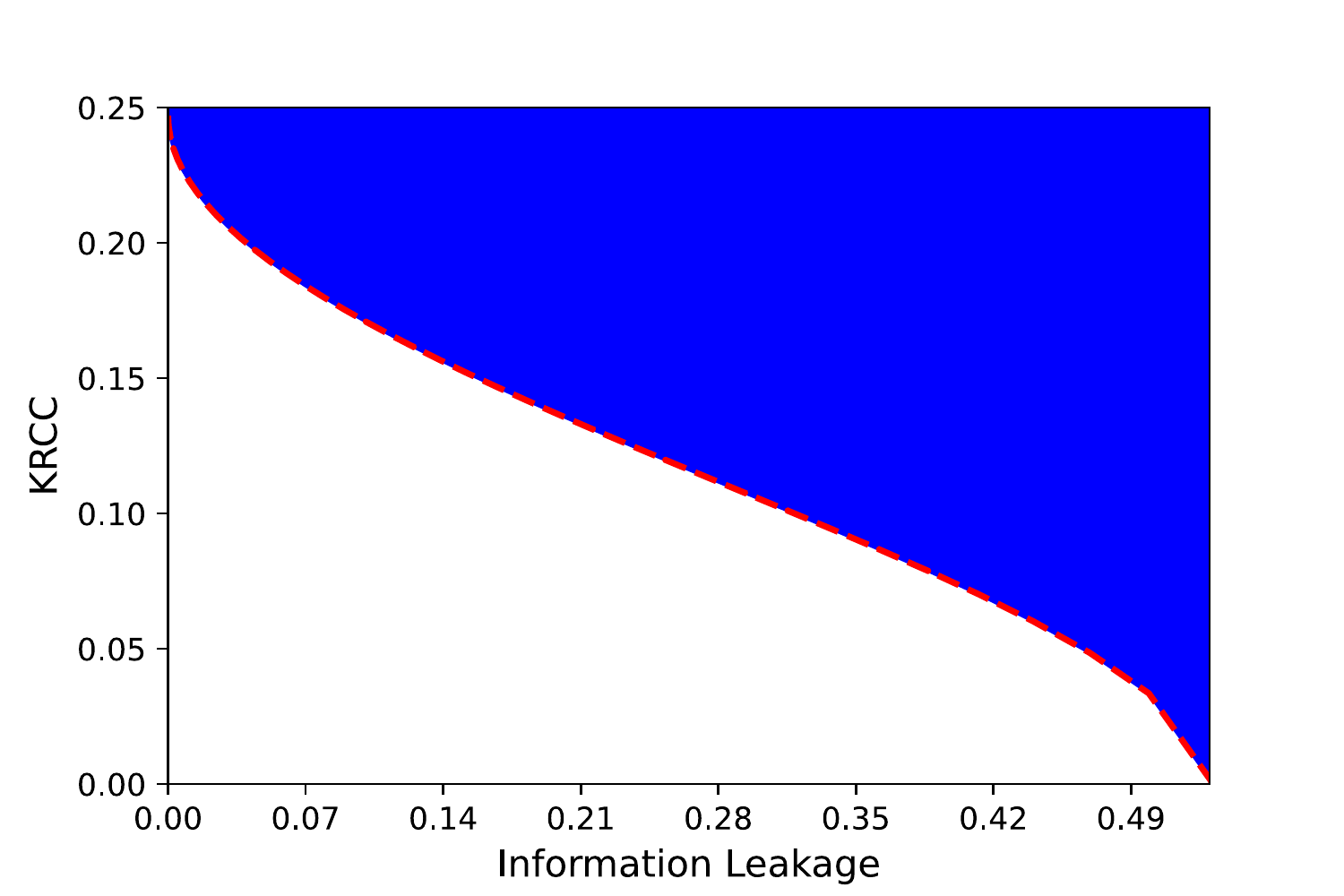}
 \caption{Privacy-Utility Tradeoff. The shaded region indicates the achievable privacy-utility set.}
 \label{fig:3}
\end{figure}

\section{Conclusion}
\label{sec:conclusion}
We have considered the privacy-utility tradeoff in dataset obfuscation, where utility is measured with respect to KRCC metric and privacy is measured as privacy leakage under fingerprinting attacks. 
We have quantified a fundamental trade-off between rank-preservation and user anonymity. We have considered
single-letter obfuscation mechanisms  and their fundamental performance limits were characterized. We have characterized the optimal obfuscating test-channel, optimizing the privacy-utility tradeoff in the form of a convex optimization problem which can be solved efficiently.



\bibliographystyle{unsrt}
\bibliography{References}

\newpage
\appendices
\section{Proof of Theorem \ref{th:1}}
\label{A:1}
Consider a fixed $n,m,q\in \mathbb{N}$, a distribution $P_X$, conditional distributions $P_{F|X}$ and $P_{Y|X}$. We first evaluate the resulting KRCC measure when a single-letter obfuscation mechanism $P_{Y|X}$ is applied to a dataset  $\mathsf{X}$ parametrized by $(n,m,\mathcal{X},P_X)$. Let $R_d(\cdot)$ and $\widehat{R}_d(\cdot)$ denote the degree-based rank function associated with the original dataset $\mathsf{X}$ and obfuscated dataset $\mathsf{Y}$, respectively. Then, 
\begin{align}
\nonumber   &\mathbb{E}(d_{KRCC}(R_d, \widehat{R}_d))
   \\&\nonumber=\mathbb{E}\bigg( \frac{1}{n(n-1)}\sum_{k,l=1}^{n} \mathbbm{1} \Big(R_d(k)>R_d(l) \!\!\!\And\!\!\! \widehat{R}_d(k)<\widehat{R}_d(l)\Big)\bigg)
   \\&\nonumber \stackrel{(a)}{=}  \frac{1}{n(n-1)}\sum_{k,l=1}^{n} P \Big(R_d(k)>R_d(l) \!\!\!\And\!\!\! \widehat{R}_d(k)<\widehat{R}_d(l)\Big)
   \\&\nonumber= \frac{1}{n(n-1)} \sum_{k,l=1}^{n} P\Big(\frac{1}{\sqrt{m}}R_d(k)\!>\!\frac{1}{\sqrt{m}}R_d(l) 
   \!\!\!\And\!\!\! \frac{1}{\sqrt{m}}\widehat{R}_d(k)\!<\!\frac{1}{\sqrt{m}}\widehat{R}_d(l)\Big)
   \\&\nonumber\stackrel{(b)}{=} \frac{1}{n(n-1)}\sum_{k,l=1}^{n} P\Big(\frac{1}{\sqrt{m}}D_\mathsf{X}(k)\!>\!\frac{1}{\sqrt{m}}D_\mathsf{X}(l) 
    \!\!\And\!\! \frac{1}{\sqrt{m}}D_\mathsf{Y}(k)\!<\!\frac{1}{\sqrt{m}}D_\mathsf{Y}(l)\Big)
   \\&\nonumber\stackrel{(c)}{=}  P\Big(\frac{1}{\sqrt{m}}D_\mathsf{X}(1)>\frac{1}{\sqrt{m}}D_\mathsf{X}(2) \!\!\!\And\!\!\! \frac{1}{\sqrt{m}}D_\mathsf{Y}(1)<\frac{1}{\sqrt{m}}D_\mathsf{Y}(2)\Big)
   \\&\nonumber=
    P\Big(\frac{1}{\sqrt{m}}\sum_{j=1}^m X_{1,j}>\frac{1}{\sqrt{m}}\sum_{j=1}^m X_{2,j}  \!\!\!\And\!\!\!
   \\&\nonumber\qquad \qquad \qquad \qquad \qquad \frac{1}{\sqrt{m}}\sum_{j=1}^m Y_{1,j}<\frac{1}{\sqrt{m}}\sum_{j=1}^m Y_{2,j}\Big)
   \\&=  P\Big(\frac{1}{\sqrt{m}}\sum_{j=1}^m X_{2,j}- X_{1,j}<0 \!\!\!\And\!\!\! \frac{1}{\sqrt{m}}\sum_{j=1}^m Y_{1,j}- Y_{2,j}<0\Big),
   \label{eq:A:1:1}
\end{align}
where (a) follows the form linearity of expectation, (b) uses the definition of degree-based rank function (Definition \ref{def1}), and (c) follows from the fact that the original dataset elements are IID and in single-letter obfuscation mechanisms the obfuscating test-channels are statistically identical. We bound the last term using a generalization of the Berry-Esseen result to multivariate scenarios given in \cite[Theorem 1.1]{raivc2019multivariate}. The theorem is stated below for completeness. 
\begin{Theorem}
\label{th:4}
[\cite{raivc2019multivariate}, Theorem 1.1]
Let $(Z_{1,i},Z_{2,i}), i\in [m]$ be independent pairs of  sequences of independent, zero-mean, and unit-variance random variables, where $m\in \mathbb{N}$, and let $W'_j\triangleq  \frac{1}{\sqrt{m}}{\sum_{i\in [m]} Z_{j,i}}, j\in \{1,2\}$. Then, for any measurable convex set $\mathcal{A}$,
\begin{align*}
|P((W'_1, W'_2) \!\in\! \mathcal{A})- P((N'_1,N'_2) \in \mathcal{A})|\leq \frac{(42\sqrt[4]{2}+16)}{m\sqrt{m}} \sum_{i=1}^{m} \mathbb{E}|\underline{Z}_i|^3,
\end{align*}
where $(N'_1,N'_2)$ is a pair of independent and unit-variance Gaussian random variables, and $|\underline{Z}_i|\triangleq \sqrt{Z_{1,i}^2+Z_{2,i}^2}, i\in[m]$.
\end{Theorem}

We let $W_1\triangleq \frac{1}{\sqrt{m}}\sum_{j=1}^m X_{2,j}- X_{1,j}$ and $W_2\triangleq \frac{1}{\sqrt{m}}\sum_{j=1}^m Y_{1,j}- Y_{2,j}$. Then, $W_1$ and $W_2$ are zero-mean variables since $X_{i,j}, i\in [n], j\in [m]$ are IID and the obfuscating test-channels are statistically identical so that  $Y_{i,j}, i\in [n], j\in [m]$ are IID.
Next, we find the covariance matrix of $(W_1,W_2)$. First, we find the variance of $W_1$:
\begin{align*}
    &Var(W_1)=\mathbb{E}(W_1^2)= \mathbb{E}\left(\left(\frac{1}{\sqrt{m}} \sum_{j=1}^m (X_{1,j}- X_{2,j})\right)^2\right)
    \\& \stackrel{(a)}{=} \frac{1}{m}\sum_{j=1}^{m}\mathbb{E} \left(\left(X_{2,j}- X_{1,j}\right)^2\right)+ \frac{1}{m}\sum_{i\neq j}\mathbb{E}\left(\left(X_{1,i}-X_{2,i})(X_{1,j}-X_{2,j}\right)\right)
    \\& 
    \stackrel{(b)}{=}
    \mathbb{E}\left(\left(X_{1,1}- X_{2,1}\right)^2\right)
    \stackrel{(c)}{=}P(X_{1,1}\neq X_{2,1})
     = P_{X_1,X_2}(0,1)+P_{X_1,X_2}(1,0)\stackrel{(d)}{=} 2P_X(0)P_X(1),
\end{align*}
where in (a) we have used the fact that the dataset elements are identically distributed and the test-channels are statistically identical,  (b) follows since $(X_{1,i},X_{2,i})$ is independent of $(X_{1,j},X_{2,j})$ since the dataset elements are IID, (c) follows since $X_{1,1}-X_{2,1}\in \{-1,0,1\}$, and (d) follows since $X_1$ and $X_2$ are IID.
Variance of $W_2$ is similarly derived as:
\begin{align*}
    & Var(W_2)= 2P_Y(0)P_Y(1)
    = 2\Big(P_X(0)P_{Y|X}(0|0)+ P_X(1)P_{Y|X}(0|1)\Big)
    \Big(P_X(0)P_{Y|X}(1|0)+ P_X(1)P_{Y|X}(1|1)\Big)
\end{align*}
The covariance between $W_1$ and $W_2$ is given by: 
\begin{align*}
    &Cov(W_1,W_2)= \mathbb{E}(W_1W_2)
    = \mathbb{E}\left(\frac{1}{\sqrt{m}} \sum_{i=1}^m (X_{2,i}- X_{1,i})) \frac{1}{\sqrt{m}} \sum_{j=1}^m (Y_{1,j}- Y_{2,j}))\right)
    \\&= \frac{1}{m} \sum_{i=1}^m\mathbb{E}\left( (X_{2,i}- X_{1,i}) (Y_{1,i}- Y_{2,i})\right)
    = \mathbb{E}\left((X_{2,1}- X_{1,1}) (Y_{1,1}- Y_{2,1})\right)
    = 2\mathbb{E}(X_{2,1}Y_{1,2})-2\mathbb{E}\left(X_{1,1}Y_{1,1} \right)
    \\&= 2P_X(1)P_Y(1)- 2P_{X,Y}(1,1)
     = 2P_X(1)(P_Y(1)-P_{Y|X}(1,1))
\end{align*}
Let $\Sigma_{W_1,W_2}$ be the covariance matrix of $(W_1,W_2)$. We define  $\underline{W}$ as the column vector consisting of $W_1,W_2$ and define $\underline{W}'\triangleq \Sigma_{W_1,W_2}^{\frac{-1}{2}}\underline{W}$. Then,
\begin{align*}
    &W'_1=\frac{1}{\sqrt{m}} \sum_{j=1}^m \Sigma_{W_1,W_2}^{\frac{-1}{2}}(1,1) (X_{2,j}-X_{1,j})+ \Sigma_{W_1,W_2}^{\frac{-1}{2}}(1,2)(Y_{1,j}-X_{2,j}),\\
    &W'_2=\frac{1}{\sqrt{m}} \sum_{j=1}^m\Sigma_{W_1,W_2}^{\frac{-1}{2}}(2,1) (X_{2,j}-X_{1,j})+ \Sigma_{W_1,W_2}^{\frac{-1}{2}}(2,2)(Y_{1,j}-X_{2,j}),\\
\end{align*}
where $\Sigma^{-\frac{1}{2}}_{W_1,W_2}(i,j), i,j\in \{1,2\}$ is the $(i,j)$th element of the matrix $\Sigma^{-\frac{1}{2}}_{W_1,W_2}$. 
It should be noted that $\Sigma_{W_1,W_2}^{\frac{-1}{2}}$ exists since $\Sigma_{W_1,W_2}$ is positive semi-definite. It is straightforward to check that $W'_1,W'_2$ are zero-mean and unit variance. 
Consequently, $W'_1,W'_2$ satisfy the properties of Theorem \ref{th:4}. Let 
\[\mathcal{A}\triangleq \{\underline{w}'\in \mathbb{R}^2: \Sigma_{W_1,W_2}^{\frac{1}{2}}\underline{w}'\leq 0\}\]
Then, by Theorem \ref{th:4}, we have: 
\begin{align*}
|P((W'_1, W'_2) \in \mathcal{A})- P((N'_1,N'_2) \in \mathcal{A})|\leq \frac{(42\sqrt[4]{2}+16)}{\sqrt{m}} \mathbb{E}|\underline{Z}|^3,
\end{align*}
where $\underline{Z}=({Z}_1,{Z}_2)$, and 
\begin{align*}
&{Z}_1= \Sigma_{W_1,W_2}^{\frac{-1}{2}}(1,1) (X_{2,1}-X_{1,1})+ \Sigma_{W_1,W_2}^{\frac{-1}{2}}(1,2)(Y_{1,1}-X_{2,1})
\\&{Z}_2= \Sigma_{W_1,W_2}^{\frac{-1}{2}}(2,1) (X_{2,1}-X_{1,1})+ \Sigma_{W_1,W_2}^{\frac{-1}{2}}(2,2)(Y_{1,1}-X_{2,1}). \end{align*}
Note that 
\[|\underline{Z}|\leq |\Sigma^{-\frac{1}{2}}_{W_1,W_2}|_F |[X_{2,1}-X_{1,1}, Y_{1,2}-{Y_{2,1}}]|,\]
where $|\cdot|_F$ is the Frobenius norm. Let $\Sigma_{W_1,W_2}= V \Lambda V^*$, where $V$ and $\Lambda$ are the singular value decomposition matrices associated with $\Sigma_{W_1,W_2}$, $V$ is unitary, and $V^*$ is the conjugate transpose of $V$. So, $\Sigma^{-\frac{1}{2}}_{W_1,W_2}=\Lambda^{-\frac{1}{2}}V^*$. So, $|\Sigma^{-\frac{1}{2}}_{W_1,W_2}|_F\leq |\Lambda^{-\frac{1}{2}}|_F|V^*|_F\leq  \frac{4}{\sqrt{\lambda^*}}$, where $\lambda^*$ is the smallest eigenvalue of $\Sigma_{W_1,W_2}$ and we have used the fact that $|V^*|_F=trace(VV^*)=trace(I_2) = 2$. Furthermore, by the Gershgorin circle theorem \cite{sahami2022fast,bhatia2013matrix},  we have $\lambda^*\geq \max(\sigma_{1,1}-|\sigma_{1,2}|, \sigma_{2,2}-|\sigma_{2,1}|)$. So, $|\Sigma^{-\frac{1}{2}}_{W_1,W_2}|_F\leq \frac{4}{\lambda^*}=\theta$. Let $\gamma\triangleq\mathbb{E}(|[X_{2,1}-X_{1,1}, Y_{1,2}-{Y_{2,1}}]|)$. Then,  $\gamma$ is given by:

\begin{align*}
    &\gamma=\mathbb{E}\Bigg(\Big( {(X_{2,1}-X_{1,1})^2+(Y_{1,1}-Y_{2,1})^2}\Big)^\frac{3}{2}\Bigg)
    \\& =P\big(|X_{2,1}-X_{1,1}|=0, |Y_{1,1}-Y_{2,1}|=1\big)
    \\& +P\big(|X_{2,1}-X_{1,1}|=1, |Y_{1,1}-Y_{2,1}|=0\big)
    \\&+ 2^{\frac{3}{2}}P\big(|X_{2,1}-X_{1,1}|=1, |Y_{1,1}-Y_{2,1}|=1\big)
    \\&=2P(X=Y)P(X\neq Y)+2^{\frac{5}{2}}(P_{X,Y}(0,0)P_{X,Y}(1,1)+P_{X,Y}(0,1)P_{X,Y}(1,0))
\end{align*}
So far, we have shown that: 
\begin{align*}
|P((W'_1, W'_2) \in \mathcal{A})- P((N'_1,N'_2) \in \mathcal{A})|\leq \frac{(42\sqrt[4]{2}+16)}{\sqrt{m}} \theta\gamma.
\end{align*}
Let $\underline{N}=[N_1,N_2]$, where $\underline{N}= \Sigma_{W_1,W_2}^{\frac{1}{2}}\underline{N}'$. It is straightforward to show that $P((N'_1,N_2')\in \mathcal{A})= F_{N_1,N_2}(0,0)=Q(P_X,P_{Y|X})$ and that $\Sigma_{N_1,N_2}=\Sigma_{W_1,W_2}$. As a result, from Equation \eqref{eq:A:1:1}, we have:
\begin{align*}
   | \mathbb{E}(d_{KRCC}(R_d,\widehat{R}_d))- Q(P_X,P_{Y|X})|\leq \frac{(42\sqrt[4]{2}+16)}{\sqrt{m}} \theta\gamma. 
\end{align*} 
Next, we evaluate the privacy cost. We have:
\begin{align*}
&I(U; F^q,\mathsf{Y})= I(U;\mathsf{Y})+I(U;F^q|\mathsf{Y})
\\&\stackrel{(a)}{=} I(U;F^q|\mathsf{Y})=
I(U,\mathsf{Y}; F^q)-I(\mathsf{Y};F^q),
\end{align*}
where we have used the chain rule of mutual information in the first and last equality, and in (a) we have used the fact that the dataset is independent of the identity of the victim (recall that the victim is chosen randomly and uniformly from the dataset members, independently of dataset elements). Furthermore, we have:
\begin{align*}
   \frac{1}{q} I(U,\mathsf{Y}; F^q)= \frac{1}{q} \sum_{j=1}^q I(Y_{U,i_j}; F_j)= I(P_Y,P_{F|Y}). 
\end{align*}
Additionally,
\begin{align*}
    I(\mathsf{Y};F^q)=D_{KL}(P_{F^q}|| P_{F^q|\mathsf{Y}})
\end{align*}
Note that $\mathsf{Y}$ is a random unstructured code with single-letter distribution $P_Y$ and hence is a 
\textit{good code} for a channel with transition probability $P_{Y|F}$ and from \cite[Theorem 7]{polyanskiy2013empirical}, we have 
\begin{align*}
 D_{KL}(P_{F^q}|| P_{F^q|\mathsf{Y}})\leq \zeta+ b\frac{\log^{\frac{3}{2}}m}{\sqrt{m}}.
\end{align*}
for some $b>0$. Consequently,
\begin{align*}
   I(P_Y,P_{F|Y})\leq  \frac{1}{q} I(U; F^q,\mathsf{Y})\leq
I(P_Y,P_{F|Y})+\zeta+b\frac{\log^{\frac{3}{2}}m}{\sqrt{m}}.
\end{align*}
This completes the proof. 
\qed

\section{Proof of Theorem \ref{th:2}}
\label{A:3}
From Theorem \ref{th:1}, we need to solve the following optimization problem:
\begin{align*}
&P^{\epsilon}_{Y|X}=\argmin_{P_{Y|X}:I(U; F^q,\mathsf{Y})<\epsilon} F_{N_1,N_2}(0,0)
\\&= \argmin_{P_{Y|X}:I(U; F^q,\mathsf{Y})<\epsilon} \int_{n_1=-\infty}^{0} \int_{n_2=-\infty}^{0} \frac{1}{2\pi \sqrt{|\Sigma|}}\exp\Big({-\frac{1}{2} \underline{n}^t \Sigma^{-1}\underline{n}}\Big)dn_1 dn_2
\end{align*}
where, $\underline{n}=\begin{bmatrix}
    n_1 \\ n_2
\end{bmatrix}$. 

Let us define $p_1\triangleq P_{Y|X}(1|0),  p_2\triangleq P_{Y|X}(0|1)$. and define  variables $N'_1= \frac{N_1}{\sqrt{2Var(N_1)}},N'_2= \frac{N_2}{\sqrt{2Var(N_2)}}$. Note that $P(N_1\leq 0, N_2\leq 0)=P(N'_1\leq 0, N'_2\leq 0)$.
The covariance of $N'_1, N'_2$ is:
 \begin{align*}
     \Sigma'=\begin{bmatrix}
    \frac{1}{2} & \frac{1}{2}-p' \\ \frac{1}{2}-p' & \frac{1}{2}\end{bmatrix}
 \end{align*} 
 where $p'$ is defined as:
\begin{align}
    & p'\triangleq \frac{1}{2}- \frac{\text{Cov}(N_1,N_2)}{2\sqrt{\text{Var}(N_1)\text{Var}(N_2)}}.
    \label{eq:pp}
\end{align}
So, 
\begin{align*}
&P^{\epsilon}_{Y|X}=\argmin_{p_1,p_2:I(U; F^q,\mathsf{Y})<\epsilon} F_{N'_1,N'_2}(0,0)
\\&= \!\!\!\!\argmin_{p_1,p_2:I(U; F^q,\mathsf{Y})<\epsilon} \int_{n'_1=-\infty}^{0} \int_{n'_2=-\infty}^{0} \frac{1}{2\pi \sqrt{|\Sigma'|}}\exp\Big({-\frac{1}{2} \underline{n'}^t \Sigma'^{-1}\underline{n'}}\Big)d\underline{n}'
\end{align*}
The eigenvalues of the covariance matrix $\Sigma'$ are $\lambda_1= 1-p', \lambda_2=p'$ and the eigenvectors are the columns of $V=\frac{1}{\sqrt{2}}\begin{bmatrix} \frac{1}{\sqrt{1-p'}} & \frac{1}{\sqrt{p'}} \\ \frac{1}{\sqrt{1-p'}} & -\frac{1}{\sqrt{p'}} \end{bmatrix}$. Let us define $\underline{n}''= V^t\underline{n}'$
Changing the variables in the integral, we have:
\begin{align*}
 &\argmin_{p_1,p_2:I(U;  F^q,\mathsf{Y})<\epsilon}\int_{n''_1=-\infty}^{0} \int_{n''_2=\sqrt{\frac{1-p'}{p'}}n''_1}^{-\sqrt{\frac{1-p'}{p'}}n''_1} \frac{1}{2\pi} \exp(-\frac{{n''_1}^{2}+{n''_2}^{2}}{2} )d_{n''_1} d_{n''_2}
\end{align*}
Form Equation \eqref{eq:pp} and using Cauchy-Schwarz inequality we have $0<p'<1$, so the inner integral interval $[\sqrt{\frac{1-p'}{p'}}n''_1,-\sqrt{\frac{1-p'}{p'}}n''_1]$ is decreasing as function of $p'$ for all values of $n''_1<0$. Hence, to minimize the utility cost, we should take the maximum value of $p'$ such that $I(U;  F^q,\mathsf{Y})<\epsilon$. That is, the optimization becomes
\begin{align*}  
\argmax_{p_1,p_2:I(P_Y,P_{Y|F})\leq\epsilon}p'=  \argmin_{p_1,p_2:I(P_Y,P_{Y|F})\leq\epsilon} \frac{\text{Cov}(N_1,N_2)}{\sqrt{\text{Var}(N_1)\text{Var}(N_2)}}.
\end{align*}
From the proof of Theorem \ref{th:1}, we have 
\begin{align*}
&Cov(N_1,N_2)= P_X(1)(P_Y(1)-P_{Y|X}(1|1))\\
&= P_X(1) (p_1P_{X}(0)+(1-p_2)P_{X}(1)-P_{Y|X}(1|1))\\
&= P_X(1)P_X(0)(p_1+p_2-1).
\end{align*}
So, the optimization can be re-written as: 
\begin{align*} P^{\epsilon}_{Y|X}\!\!=\!\!  \!\!\!\!\argmin_{p_1,p_2:I(P_Y,P_{Y|F})\leq\epsilon} \!\!\!\frac{p_1+p_2-1}{\sqrt{(\overline{p}_1P_X(0)+p_2P_X(1))(p_1P_X(0)+\overline{p}_2P_X(1))}},
\end{align*}
where $\overline{p}_1\triangleq 1-p_1$ and $\overline{p}_2\triangleq 1-p_2$. 
Note that we can restrict the minimization to $p_1+p_2<1$ since if $p_1+p_2>1$ the objective function in the above optimization is positive and cannot achieve the minimum value since for $p_1+p_2=0$ the value of zero is already achieved. So, 
\begin{align*} P^{\epsilon}_{Y|X}\!\!=\!\!  \!\!\!\argmin_{\substack{p_1,p_2:I(P_Y,P_{Y|F})\leq\epsilon\\p_1+p_2<1}}\!\!\! \frac{p_1+p_2-1}{\sqrt{(\overline{p}_1P_X(0)+p_2P_X(1))(p_1P_X(0)+\overline{p}_2P_X(1))}}.
\end{align*}

Next, we show that $\frac{\text{Cov}(N_1,N_2)}{\sqrt{\text{Var}(N_1)\text{Var}(N_2)}}$ is i) increasing in $p_1$ for any fixed $p_2$, and ii) increasing in $p_2$ for any fixed $p_1$. Hence, we conclude that the optimal value is achieved at the boundary where $I(P_Y,P_{Y|F})=\epsilon$ since for any point not on the boundary, one either reduce $p_1$ or $p_2$ without violating  $I(P_Y,P_{Y|F})\leq\epsilon$. To show i) it suffices to show that  $\frac{\text{Cov}^2(N_1,N_2)}{{\text{Var}(N_1)\text{Var}(N_2)}}$ is  decreasing in $p_1$ for any fixed $p_2$ since $\text{Cov}^2(N_1,N_2)<0$ for $p_1+p_2<1$. That is, we wish to show that the following function is decreasing in $p_1$ for fixed $p_2$: 
\begin{align*}
&\frac{(p_1+p_2-1)^2}{(\overline{p}_1P_X(0)+p_2P_X(1))(p_1P_X(0)+\overline{p}_2P_X(1))}
= \frac{p_1+p_2-1}{\overline{p}_1P_X(0)+p_2P_X(1)}\frac{p_1+p_2-1}{p_1P_X(0)+\overline{p}_2P_X(1)}. 
\end{align*} 
Taking derivative of each component in the multiplication with respect to $p_1$, we have:
\begin{align*}
  &  \frac{\partial}{\partial p_1}\frac{p_1+p_2-1}{\overline{p}_1P_X(0)+p_2P_X(1)}
  = \frac{\overline{p}_1P_X(0)+p_2P_X(1)+P_X(0)(p_1+p_2-1)}{(\overline{p}_1P_X(0)+p_2P_X(1))^2}
  = \frac{p_2}{(\overline{p}_1P_X(0)+p_2P_X(1))^2}
\end{align*}
and
\begin{align*}
  &  \frac{\partial}{\partial p_1}\frac{p_1+p_2-1}{p_1P_X(0)+\overline{p}_2P_X(1)}= \frac{p_1P_X(0)+\overline{p}_2P_X(1)-P_X(0)(p_1+p_2-1)}{(p_1P_X(0)+\overline{p}_2P_X(1))^2}
  = \frac{\overline{p}_2}{(\overline{p}_1P_X(0)+p_2P_X(1))^2}
\end{align*}
Since both derivatives are positive, and both functions are negative-valued, the multiplication has a derivative which is negative with respect to $p_1$ for fixed $p_2$. So, $\frac{\text{Cov}^2(N_1,N_2)}{{\text{Var}(N_1)\text{Var}(N_2)}}$ is  decreasing in $p_1$ for any fixed $p_2$. Hence,  $\frac{\text{Cov}(N_1,N_2)}{\sqrt{\text{Var}(N_1)\text{Var}(N_2)}}$ is  increasing in $p_1$ for any fixed $p_2$ which proves i). The proof of ii) follows similarly. We conclude that the optimum is achieved at $I(P_Y,P_{Y|F})=\epsilon$.
This completes the proof. 
\qed

\end{document}